\documentclass[12pt, oneside]{article}
\usepackage[utf8]{inputenc}
\usepackage[english]{babel}
\usepackage{amsfonts}
\usepackage{amsmath}
\usepackage{amsthm}
\usepackage{MnSymbol}
\usepackage{wasysym}
\usepackage{hyperref}
\usepackage{mathtools}
\usepackage{multicol}
\usepackage[letterpaper,margin=2.5cm,left=3cm,right=2.5cm, includehead]{geometry}
\usepackage{graphicx}
\usepackage{float}
\usepackage{enumitem}
\usepackage{setspace}
\usepackage{mathrsfs}
\usepackage{pdfpages}
\usepackage{titling}
\usepackage{tocbibind}
\usepackage[spanish]{babelbib}
\usepackage{xcolor}

\usepackage{eso-pic}

\newcommand{\htau}{\hat{\tau}}

\newcommand{\calT}{\mathcal{T}}

\newtheorem{theorem}{Theorem}

\newtheorem{lemma}{Lemma}
\newtheorem{proposition}{Proposition}
\newtheorem{corollary}{Corollary}
\newtheorem{example}{Example}
\newtheorem{remark}{Remark}

\begin{document}

\title{The causal structure of the c-completion of warped spacetimes}
\author{Luis Aké Hau, Saul Burgos and Didier A. Solis}
\maketitle

\begin{abstract}
In this work we study the structure of the  future causal completion $\hat{M}$ of a globally hyperbolic GRW spacetime $\mathbb{R}\times_\alpha M$ using the novel notion of Lorentzian pre-length spaces. As our main result, we prove that the causal completion of a GRW spacetime is a globally hyperbolic pre-length space provided the chronological topology is Hausdorff.
\end{abstract}

\section{Introduction}

Understanding the ultimate fate of the universe has been one of the driving forces behind cosmology. In the modern context of general relativity this translates to the study of a structure that encodes the information carried along by massive observers and massless particles,  such as photons, whose wordlines are represented, respectively, by (future) timelike and lightlike inextendible curves in spacetime. One of the most successful alternatives for modeling the end of spacetime is the \textit{conformal boundary} first introduced by R. Penrose \cite{zero}. This construction is inspired in the classical set up of projective geometry where ideal points (or ``points at infinity") are attached at the end of inextendible curves. This is achieved by considering a conformal embedding  $i:(M,g)\to (\bar{M},\bar{g})$, $i^*\bar{g}={\Omega}^2g$ of the spacetime $(M,g)$  in a way that $\partial i(M)=\Omega^{-1}(0)$. Due to the conformal invariance of lightlike geodesics under conformal changes of the metric, those future lightlike geodesics of $(M,g)$ that acquire a future endpoint in $\partial i(M)$ are indeed future complete, and hence its endpoints can be considered as points at infinity. Among the most remarkable features of the causal boundary is that it allows to provide a sound mathematical definition of black hole regions \cite{HE}. Although elegant from the mathematical point of view, this approach has the serious drawback of not being intrinsic, as it requires of a concrete embedding in order to construct the conformal boundary $\mathfrak{I}:=\partial i(M)$. As a consequence, to this day there does not exist a general criterion in order to know if a given spacetime admits a conformal boundary, or a canonical way to construct such and embedding in case it exists.

A different alternative to explore the end of spacetime was developed by R. Geroch, E. Kronheimer and R. Penrose \cite{gkp}. The so called \textit{future causal completion} $\hat{M}$ of a spacetime $(M,g)$ is constructed by assigning an ideal point to all observers having a common past, and hence sharing the same observer horizon. In a time dual way a past causal completion $\check{M}$ is defined. Finally, the \textit{causal boundary} arises as a quotient of $(\hat{M}\cup \check{M})\setminus M$, when a series of non-trivial identifications take place (see for instance \cite{flores:revisited,onthefinal} and references therein for detailed accounts on the necessity of such identifications). The causal boundary, besides being manifestly intrinsic --as it only depends on the causal structure of $M$-- it also coincides, under mild additional assumptions, with the conformal boundary, when the latter exists \cite{onthefinal}. Moreover, it also carries enough structure to provide a formalism for the study of black holes in a general setting \cite{bhcausal}.

Causal completions do not have a manifold structure, thus they can not be furnished with a spacetime structure in the classical way. However, they can be endowed with diffe\-rent topologies \cite{beemclt,harrischrtop,floresharris,clt,Olaf}, a chronological structure \cite{harrisuniversality} and even linear co\-nnections \cite{harriscon}. In very recent times there has been a great interest to explore alternative frameworks for general relativity that do not rely on $C^2$ Lorentzian manifolds. Commonly termed as non regular geometric approaches, they have provided solid grounds for generalizing some of the landmark features of relativity, such as causality \cite{mingucone,piotrc0,ling01} and singularity theorems \cite{Graf3,Kun3}. Among these alternatives there are some whose formulations do not require a metric tensor at all \cite{KSlls,lorentzmetric}, and hence are suitable candidates to study  causal completions and bring a better understanding on observational data corresponding to scenarios --like black hole merging and gravitational waves \cite{ebh,ligo}-- in which $C^2$ regularity might not be guaranteed. In fact, in \cite{complls} it was shown that the future causal completion of a globally hyperbolic spacetime admits a Lorentzian pre-length space structure, as defined in \cite{KSlls}.

 The globally hyperbolic stage of the causal ladder of spacetimes is particularly important in General Relativity as it is the natural environment for the global well posedness of the Einstein equations, due to Choquet Bruhat and Geroch \cite{choque-bruhat, choquet-bruhat_geroch}; and is also an essential assumption in the Singularity Theorems of Hawking and Penrose \cite{HE, penrose_sing}. In Lorentzian length spaces, global hyperbolicity has also proven to be important. A series of important results and problems can be extended to this generalized setting: the Lorentzian splitting theorem \cite{splitting}, time functions and the topological splitting by Geroch \cite{timefunctionsLLS}, and Bartnik's splitting conjecture \cite{bartnik}. Moreover, it is the natural place to study convergence of Lorentzian spaces \cite{mondino-samann}. 

In this work we specialize the findings of \cite{complls} to a cosmological setting. Thus we explore in greater detail the causal structure of the future causal completion of generalized Robertson-Walker spacetimes, thus establishing the basic features for the particular case of FLRW spacetimes. This work is organized as follows: in Section \ref{sec:pre}  we introduce the notation and terminology that will be used throughout this work, as well as the basic notions pertaining Lorentzian pre-length spaces. In Section \ref{sec:CLpLS} we address the structure of the future causal completion of a GRW spacetime. Finally, in Section \ref{sec:causality} we describe its causal properties and prove that it is a globally hyperbolic Lorentzian pre-length space when its chronological topology is Hausdorff.

\section{Preliminaries}\label{sec:pre}

\subsection{Causal completions}

Let us consider a strongly causal spacetime $(M,g)$, that is, a connected time oriented $n$ dimensional Lorentzian manifold. We define the \textit{chronological} (\textit{causal}) relation $\ll$ $(\le )$ as $p\ll q$ $(p\le q)$ if there exists a future timelike (\textit{causal}) curve from $p$ to $q$\footnote{By convention we set $p\le p$, for all $p\in M$.}, and the \textit{chronological} (\textit{causal}) sets accordingly:
\begin{eqnarray*}
I^{+}(p) =\{q \in M \mid p \ll q\}, &\qquad& J^{+}(p)=\{q \in M \mid p \leq q\},\\
I^{-}(p) =\{q \in M \mid q \ll p\}, &\qquad& J^{-}(p)=\{q \in M \mid q \leq p\}.
\end{eqnarray*}
Further, we say that a sequence of points $\{x_{n}\}$ is a \textit{future-directed chain} if $x_{n} \ll x_{n+1}$ for all $n\in\mathbb{N}$. \textit{Past directed chains} are defined time dually.  An \textit{inextensible chain} is characterized by its non convergence with respect to the manifold topology.

 A subset $P \subseteq M$ is called \textit{past set}  if it coincides with its chronological past, that is,  $P=I^{-}(P)$. Such a set is \textit{indecomposable} (or \textit{IP} for short) if $P$ can not be expressed as the union of two proper past subsets of $P$. A classical result (see Thrms. 2.1 and 2.3 in \cite{gkp}) establishes that there are two mutually exclusive classes of IPs: the ones of the form $I^-(p)$, $p\in M$, called \textit{proper} (or \textit{PIP}) and those corresponding to the chronological past of an inextensible future chain $\{x_n\}$, called \textit{terminal}  (or \textit{TIP}). The \textit{future causal completion} $\hat{M}$ is then defined as the set of all IPs. Notice that the sets of PIPs can be naturally identified with $M$ while the TIPs correspond to ideal points, and thus conform the \textit{future causal boundary} $\hat{\partial} M$. Hence we have the following decomposition:
\[	
IPs \equiv PIPs \cup TIPs, \quad	
\hat{M}\equiv  M \cup \hat{\partial}M.
\]
Taking time duals allows us to define the past causal boundary $\check{M}$ and leads to the identification $\overline{M}=M\cup \check{\partial}M$. In order to build a causal completion out of $\tilde{M}=\hat{M}\cup \check{M}$ we first note that $\tilde{M}$ contains two copies of $M$ that ought to be identified. Actually, more non-trivial identifications have to take place in order to avoid inconsistencies and ultimately define the so-called \textit{causal boundary} of $M$. The analysis required for these identifications is very subtle and we will not pursue it here, as we will focus only on the future causal boundary $\hat{M}$\footnote{For globally hyperbolic spacetimes these identifications become trivial \cite[3.29]{onthefinal}.}. A thorough discussion on the topic and can be found in \cite{onthefinal}.

\subsection{Lorentzian pre-length spaces}

The quest of establishing an axiomatic approach to causality has been present since the early years of mathematical relativity. Indeed, one of the first systematic attempts to axiomatize the most basic aspects of causality is the theory of \textit{causal spaces} due to E. Kronheimer and R. Penrose \cite{kpcs},  that serves as foundation to different developments, including some approaches to quantum gravity (see \cite{surya} and references therein). More recently, the growing interest of finding synthetic geometrical methods to deal with non-regular scenarios has led to the introduction of different frameworks \cite{lorentzmetric,Olaf2}. In particular, the notion of Lorentzian pre-length space as introduced by M. Kunzinger and C. S\"amman has proved very useful in the analysis of singularities \cite{KSlls}, extendibility of spacetimes \cite{GKSinext}, and most remarkably in setting the basis for synthetic Lorentizian comparison geometry \cite{solismontes,beran01} and Lorentzian geometric measure theory \cite{Mcann}. Among this very active area of research we have applications to contexts related to relativity \cite{complls,ACS} as well as contact geometry \cite{hendike} and hyperspace theory \cite{hyper}.

A Lorentzian pre-length space  $(X,d,\ll,\leq, \tau)$ is a metric space $(X,d)$ furnished with two relations $\ll$, $\leq$,  and a function $\tau: X\times X\to [0,\infty]$ that satisfy the following axioms:
\begin{enumerate}
\item $\leq$ is a pre-order,
\item  $\ll$ is a transitive relation contained in $\leq$,
\item $\tau$ is lower semi-continuous  ---with respect to $d$--- satisfying
\begin{itemize}
    \item $\tau(x,z)\geq \tau(x,y) + \tau(y,z)$ for all $x\leq y \leq z$,
    \item $\tau(x,y)>0$ if and only if $x\ll y$.
\end{itemize}
\end{enumerate}
In this context, the relations  $\ll$ and $\le$ are called \textit{chronological} and \textit{causal}, respectively, while $\tau$ is the \textit{time separation function}. In fact every smooth  spacetime is a Lorentzian pre-length space with respect to a metric $g$ that generates its manifold topology and $\ll$, $\le$, $\tau$ its standard relations and time separation (see Example 2.11 in \cite{KSlls}).

As easy consequences of the definition we have that the chronological sets $I^+(p)$, $I^-(p)$ are open, in accordance to the classical smooth Lorentzian setting. In addition, the \textit{push up} property holds: 
\[ 
x\leq y\ll z\ \textrm{or}\ x\ll y\leq z\ \textrm{implies}\ x\ll z.
\]
Lorentzian pre-length spaces provide enough structure to build an abstract causal theory that mimic the classical causality of smooth spacetimes \cite{KSlls}. For instance, we say that a Lipschitz curve $\gamma : I\subset\mathbb{R}\to (X,d)$ is \textit{future timelike} (\textit{causal}) if $t<s$ $(t\le s)$ implies $\gamma (t)\ll \gamma (s)$ $(\gamma (t)\le \gamma (s))$ for all $t,s\in I$. Past timelike and causal curves are defined analogously.

In particular, we say that a Lorentzian pre-length space $(X,d,\ll,\leq, \tau)$ is \textit{globally hyperbolic} if the following two properties are met
\begin{enumerate}
    \item The causal diamonds $J^+(p)\cap J^-(q)$ are compact for every $p\le q$.
    \item For every compact subset $K\subset X$, any causal curve contained in $K$ has bounded $d$-length.\footnote{This notion is commonly referred to as \textit{non total imprisonment}.} 
\end{enumerate}
A set $A\subset X$ is called \textit{causally convex} if any causal curve with endpoints in $A$ is fully contained in $A$. Followiing \cite{Olaf}, we say that $X$ is \textit{almost strongly causal} at $p\in X$ if any neighborhood $U$ of $p$ contains a causally convex neighborhood $V\ni p$. Recall that for smooth spacetimes this notion is equivalent to the fact that the Alexandrov sets $\{I^+(x)\cap I^-(y)\}$ generate the manifold topology. However, in the Lorentzian pre-length space setting the latter notion is stronger than the former.

\section{Causal completions as Lorentzian pre-length spaces}\label{sec:CLpLS}

Now we describe a way to provide the future causal completion $\hat{M}$ of a globally hyperbolic smooth spacetime $(M,g)$ with a Lorentzian pre-length structure that naturally extends its spacetime causality. The first step consists in defining relations $\hat{\ll}$ and $\hat{\le}$ on $\hat{M}$ as follows\footnote{This relation was suggested by S. Harris in \cite{harrischrtop}.}:
 	\begin{align*}\label{eq2}
	P \hat{\ll} Q &\Leftrightarrow \exists q \in Q\setminus P {\text{ such that }} P \subset I^{-}(q)  \\
	P \hat{\leq} Q & \Leftrightarrow P \subset Q.
	\end{align*}
Note that if $P \hat{\ll} Q \hat{\ll} R$ then there exist $q \in Q \setminus P$ and $r \in R \setminus Q$ with $P \subset I^{-} (q) \subset Q \subset I^{-}(r)$, thus $P \hat{\ll} R$. Hence $\hat{\ll}$ is transitive. The remaining properties in the definitions of chronological and causal structure follow immediately.

\begin{example}\label{lemma:j+compcausallyconvex}
Notice that the sets $(\hat{J}^{\pm} (P))^{c}$, where the complement is taken in $\hat{M}$, are causally convex for every $P \in \hat{M}$. Indeed, let $P \in \hat{M}$ and $\gamma : [a,b] \to \hat{M}$ be a future causal curve with $\gamma (a), \gamma (b) \in (\hat{J}^{\pm} (P))^{c}$. Proceeding by contradiction assume that there exists $c \in (a,b)$ with $\gamma (c) \in \hat{J}^{\pm} (P)$. Since $\gamma$ is future causal, then
$P \hat{\leq} \gamma (c) \hat{\leq} \gamma (b)$, which implies $ \gamma (b) \in \hat{J}^{\pm} (P)$. In a similar fashion, the sets of the form  $\hat{I}^{+}(P) \cap (\hat{J}^{-} (Q))^{c}$ 
are causally convex.
\end{example}

The limit operator $\hat{L}_{chr}$ on the sequence  of IPs given by
\begin{equation}
\label{eq3}
P \in \hat{L}_{chr}(\{P_{n}\}) \Leftrightarrow P \subset LI(\{P_{n}\}) \text{ and it is maximal in } LS(\{P_{n}\}).\footnote{The set theoretical inferior and superior limits of subsets are defined as $LI(\{P_{n}\})=\bigcup_{n=1}^{\infty}\bigcap_{m=n}^{\infty} P_{m}$ and $LS(\{P_{n}\})=\bigcap_{n=1}^{\infty} \bigcup_{m=n}^{\infty} P_{m} $, respectively.}
\end{equation}
defines the \textit{future chronological topology} $\hat{\mathcal{T}}_{chr}$ by its closed subsets as follows: a subset $C \subset \hat{M}$ is closed if and only if  $\hat{L}_{chr}(\sigma) \subset C$ holds for any sequence $\sigma \subset C$. Notice that the chronological sets $\hat{I}^\pm (P)$ are open in the $\hat{\mathcal{T}}_{chr}$ topology. There are examples of globally hyperbolic spacetimes for which the topology $\hat{\mathcal{T}}_{chr}$ is not Hausdorff, and hence the topological space $(\hat{M},\hat{\mathcal{T}}_{chr})$ is not metrizable.  For such an example refer to Harris' ``unwrapped grapefruit on a stick" spacetime \cite{floresharris}.

On the other hand, the so called \textit{closed limit topology} $\hat{\mathcal{T}}_{c}$ (or CLT for short) is metrizable. First introduced by J. Beem \cite{beemclt} and later extensively studied by I. Costa, J. Flores and J. Herrera \cite{clt}, this topology is defined in terms of the open Hausdorff limit operator
\begin{equation}
\hat{L}_{H}(\{P_{n}\})=\{P \in \hat{M} \mid P=LI(\{P_n\})=LS(\{P_{n}\})\}
\end{equation}
by defining the closed sets in $\hat{\mathcal{T}}_{c}$ as follows: a subset $C \subset \hat{M}$ is closed if and only if $\hat{L}_{H}(\{P_{n}\}) \subset C$ for every sequence $\{P_{n}\} \subset C$.

In general, $\hat{\mathcal{T}}_{chr}\subset \hat{\mathcal{T}}_{c}$, however, they coincide when their corresponding limit operators agree. Moreover, this occurs if and only if $\hat{\mathcal{T}}_{chr}$  is Hausdorff (see \cite[Thrm. 5.3]{clt}). The next result summarizes the main properties pertaining the CLT topology in globally hyperbolic spacetimes (refer to \cite[Thrms. 4.1, 4.2]{clt}).
\begin{theorem}\label{teo:clt}
If $(M,g)$ is a globally hyperbolic spacetime, then the following statements hold for the topological space $(\hat{M},\hat{\mathcal{T}}_{c})$:
\begin{itemize}
	\item[(i)] The natural inclusion $i: (M,g) \rightarrow (\hat{M},\hat{\mathcal{T}}_{c})$ given by $i(p)=I^{-}(p)$ is an open continuous map. Moreover, $i(M)$ is an open dense subset of $\hat{M}$, the induced topology on $M$ is the manifold topology, $\hat{\partial} M$ is closed and $(\hat{M},\hat{\mathcal{T}}_{c})$ is second countable.
	\item[(ii)] The chronological sets $\hat{I}^{\pm}(P)$ are open subsets for all $P \in \hat{M}$.
	\item[(iii)] Any future directed chain $\{P_{n}\} \subset \hat{M}$ converges in $\hat{\mathcal{T}}_{c}$
		\item[$(iv)$] The topological space $(\hat{M},\hat{\mathcal{T}}_{c})$ is metrizable.
\end{itemize}
\end{theorem}

It is important to notice that the metric $d_c$ that induces $\hat{\calT}_c$ can be expressed in terms of a non-atomic strictly positive finite Radon Borel measure as $\mu$ as\footnote{Here $A\triangle B$ stands for the symmetric set difference $(A\setminus B)\cup (B\setminus A)$.}
\[
d_c(A,B)=\mu (A\triangle B),
\]
The above was originally claimed without proof by Beem himself (see \cite[Thrm. 8]{beemclt}). A detailed analysis carried out in \cite{Olaf} establishes this fact and provides a refined description of $(\hat{X},d_c)$ as an intrinsic and proper metric length space \cite[Thrms. 8 and 10]{Olaf}. 

\begin{remark}\label{rem:instrinsic}
Actually, if $A\hat{\leq} B\hat{\leq}C$ then $d_c(A,C)=d_c(A,B)+d_c(B,C)$, thus any future causal curve joining $A$ to $C$ is a $d_c$ distance realizer.
\end{remark}

Thus in order to endow the future causal completion $\hat{M}$ of a globally hyperbolic spacetime with a Lorentzian pre-length space structure, we only require to define a time se\-paration $\hat{\tau}$ that is lower semi-continuous with respect to the metric $d_c$ that generates the topology $\hat{\mathcal{T}}_{c})$. This can be accomplished by considering the following function (see \cite[Thrm. 12]{complls}).
\[
\hat{\tau} (P,Q)=
\begin{cases}
0 & P\in \hat{\partial}M,\\
\sup_n \{ \tau (p,q_n) \} & P=I^{-}(p),\\
\end{cases}
\]
where $\{q_n\}$ is a timelike chain generating $Q=I^{-}(\{ q_n \} )$.

We now state a couple of results that will be used extensively in the following sections. The first one is standard  (see for example \cite[Prop. 2.11]{clt} or \cite[Prop. 4]{complls}).

\begin{lemma}\label{lemadesigualdad}
Let $(M,g)$ be is a globally hyperbolic spacetime. If $P \in \hat{\partial} M$ then for any $p\in M$ we have $P\hat{\nleq} I^-(p)$.  Equivalently, if  $P \in \hat{\partial} M$ and $P \hat{\leq} Q$, then $Q \in \hat{\partial} M$. 
\end{lemma}

\begin{lemma}\label{lemma:i+causallyconvex}
The sets $\hat{I}^{\pm} (P)$ are causally convex (when non-empty) for every $P  \in \hat{M}$.
\end{lemma}

\begin{proof}

 This is a direct consequence of the push-up property: $P \ll x \leq y$ implies $P \ll y$.
\end{proof}

\begin{lemma}\label{lemma:lcc}
Let $(M,g)$ be a globally hyperbolic space-time. Then $(\hat{M}, d_c,\hat{\ll},\hat{\leq}, \htau)$ is locally causally closed.
\end{lemma}

\begin{proof}
Let $R \in \hat{M}$ and $U$ be a neighbourhood of $R$. Consider two sequences $\{P_n\},\{Q_n\} \subset U$ converging to $P$ and $Q$ in $\overline{U}$ with the CLT topology, respectively, with $P_n \hat{\leq} Q_n$ for every $n$ and. By definition,
$$P = \limsup (\{P_n\}) = \liminf (\{P_n\}), \qquad Q = \limsup (\{Q_n\}) = \liminf (\{Q_n\}),$$
then, since $ P_n \subset Q_n$ for every $n$ we have
$$ P = \liminf (\{P_n\}) \subset \liminf (\{Q_n\}) = Q,$$
which means $P \hat{\leq} Q$.
\end{proof}

\begin{corollary}\label{coro:causally_simple}
The causal sets $\hat{J}^{\pm} (P)$ are closed for every $P$ in $\hat{M}$. In particular $\hat{M}$ is causally simple. 
\end{corollary}

\begin{proof}
Since $\hat{J}^{+} (P) \subset \overline{\hat{J}^{+} (P)}$ we only need to prove that $\overline{\hat{J}^{+} (P)} \subset \hat{J}^{+} (P)$.

Let $Q \in \overline{\hat{J}^{+} (P)}$. Then there exists a sequence $\{Q_k\} \subset \hat{J}^{+} (P)$ that converges to $Q$. That is, $P \hat{\leq} Q_k$ for every $k$. Since $\hat{M}$ is causally closed and $\{Q_k\}$ converges tu $Q$, we get $P \hat{\leq} Q$. Thus $\hat{J}^{+} (P) = \overline{\hat{J}^{+} (P)}$. The past case is analogous.
\end{proof}

\subsection{Brief review of GRW spacetimes} 

A {\em generalized Robertson-Walker spacetime} (abbrev. GRW spacetime) is a Lorentzian manifold $M=(a,b) \times S$  furnished with a metric of the form
\[
g=-dt^{2} + \alpha(t) h,
\]
where $(a,b)\subset\mathbb{R}$ is a (possibly unbounded) interval, $(S,h)$ is a Riemannian manifold and $\alpha$ is a smooth positive function over $(a,b)$. 

Generalized Robertson-Walker spacetimes include open portions of all Lorentzian spaceforms as well as the important family of FLRW spacetimes. Thus they are of uttermost importance both from the mathematical and the physical point of view.

The chronological relation on GRW spacetimes  can be characterized in the following way (see \cite[Section 2.2]{joniluis}): given $(t_{0},x_{0}), (t_{1},x_{1})\in (a,b) \times S$, 
			\[
				(t_{0},x_{0}) \ll (t_{1},x_{1}) \Leftrightarrow d(x_{0},x_{1}) < \int_{t_{0}}^{t_{1}} \frac{ds}{\sqrt{\alpha(s)}.}
			\]
Here $d$ denotes the distance induced in $S$ by the metric tensor $h$. Moreover, the future causal completion can be characterized depending on the value of $\int_{c}^{b} \frac{ds}{\sqrt{\alpha(s)}}$ for some $c \in (a,b)$. Indeed, in \cite{valana} a thorough description of $\hat{M}$ is discussed. Here we summarize the most relevant results.
\begin{theorem}\label{teo:estM}
Let $M=(a,b) \times S$ with $g=-dt^2+\alpha(t) h$ be a warped spacetime and $(S,d)$ a locally compact metric space. Then, 
\begin{enumerate} \setlength\itemsep{1em}				
\item If $\int_{c}^{b} \frac{ds}{\sqrt{\alpha(s)}}= \infty$ ,then, the future causal boundary $\hat{\partial} M$ is an infinite null cone with base $\partial_{\mathcal{B}}S \setminus \partial_{C} S$ with apex in $i^{+}$ and timelike lines over each point in $\partial_{C} S$ and final point in $i^{+}$. Moreover, $\hat{M}$ is homeomorphic to $M \cup ((a,b) \times \partial_{C} S) \cup ((a,b) \times \partial_{\mathcal{B}} S) \cup i^{+}$.
\item If $\int_{c}^{b} \frac{ds}{\sqrt{\alpha(s)}} < \infty$, then, $\hat{\partial} M$ is a copy of the Cauchy completion $S^{C}$ of $(S,h)$ and timelike lines over each point in $\partial_{C} S$ that finish in the same point at the copy at infinity of $S^{C}$. Moreover, $\hat{M}$ is homeomorphic to $M \cup ((a,b) \times \partial_{C} S) \cup (\{b\} \times S^{C})$. 
\end{enumerate}
\end{theorem}
Here $\partial_\mathcal{B} S$ and $\partial_{C} S$ denote the Busemann and (metric) Cauchy boundaries of $S$. Let us recall that the Cauchy completion $C(X)$ of a metric space $(X,d)$ consists of all the Cauchy sequences $\{x_n\} \subset X$. The Busemann completion of a Riemannian manifold $(S,h)$ is defined as follows: consider a piecewise smooth curve $c: [a, \Omega] \to S$, $-\infty < a < \Omega \leq \infty$, with $|\dot{c}|^2 = h (\dot{c}, \dot{c}) < 1$ and consider the associated function
The {\em Busemann function} of a curve $c \in C(S)$ is
$$b_c(x) = \lim_{s \nearrow \Omega} (s - d^h (x, c(s)) \in \mathbb{R} \cup \{+\infty\}, \quad x\in S.$$
Denote by $\mathcal{B} (S)$ the set of all finite Busemann functions, which is invariant under the additive action: if $b_c \in \mathcal{B}(S)$ then $b_c + k \in \mathcal{B}(S)$ for all $k \in (a-\Omega, \infty)$. Define the {\em Busemann boundary} as the quotient
$$\partial_B (S) := \mathcal{B} (S) / (a,\infty).$$
Note that $\partial_B S$ includes two types of elements, those associated to inextensible curves $c$ with $\Omega = \infty$, which can be interpreted as infinity directions in $(S,h)$; and those associated to inextensible curves $c$ with $\Omega < \infty$, which define points in the Cauchy boundary $\partial_C S$.\footnote{For a thorough description and analysis of these completions and their relation with the causal boundary of GRW spacetimes see \cite{GromovCauchy}.}
In particular, if $(S,h)$ is complete, then $\partial_{C} S=\emptyset$. This is the case when $M = I \times_{\alpha} S$ is globally hyperbolic \cite{beem}.

\section{Causality of the causal completion}\label{sec:causality}

In this section we focus our attention on the causality of the causal completion $\hat{M}$ of a globally hyperbolic GRW spacetime $(a,b)\times_\alpha S$. Our main result establishes the global hyperbolicity of $\hat{M}$. We divide our analysis in two cases, depending on the finiteness of $\int_{c}^{\infty} \frac{ds}{\sqrt{\alpha (s)}}$.

\subsection{Case A. $\int_{c}^{\infty} \frac{ds}{\sqrt{\alpha (s)}} < \infty$.}

In virtue of Theorem \ref{teo:estM}, we have that  $\hat{M} \equiv M \cup (\{b\} \times S)$ and hence the causal boundary consists of a copy of $S$ at $b$. Note that any TIP $P \in \hat{\partial} M$ can be identified with a set of the form $P = I^{-}(b, x)$, $x \in S$, where 
\[I^{-}(b,x):=\{(t_{0},x_{0})  \in M \mid \int_{t_{0}}^{b} \frac{ds}{\sqrt{\alpha(s)}}< \int_{c}^{b} \frac{ds}{\sqrt{\alpha(s)}} - d(x_{0},x) \}.
\]

First we show that almost strong causality holds on the whole $\hat{M}$.

\begin{proposition}\label{prop:weakstrongcausality}
Let $M = (a,b) \times_\alpha S$ be a globally hyperbolic GRW spacetime with $\int_{c}^{\infty} \frac{ds}{\sqrt{\alpha (s)}} < \infty$. Then $(\hat{M}, d_c, \hat{\ll},\hat{\leq},\htau)$ is almost strongly causal in any $P \in \hat{M}$.
\end{proposition}

\begin{proof}
If $P \in i (M)$ then strong causality of $M$ implies immediately the result. Let $P \in \hat{\partial} M$ and $U$ be an open neighbourhood of $P$. Let $\{p_n\}$ be a future chain generating $P$ then the sequence of IPs given by $\{I^{-} (p_n)\}$ converges to $P$ with respect to CLT, which coincides with the chronological topology in this case. We know that $\hat{I}^{+} (I^{-} (p_n))$ is a causally convex open neighbourhood of $P$ for every $n$ since $I^{-} (p_n) \hat{\ll} P$. We want to show that for some $n$ large enough, $\hat{I}^{+} (I^{-} (p_n)) \subset U$. Proceeding by contradiction we will assume that $\hat{I}^{+} (I^{-} (p_n)) \not\subset U$.

For every $n$, there exists $R_n \in \hat{I}^{+} (I^{-} (p_n))$ with $R_n \notin U$, and by definition there exists $r_n \in R_n \setminus I^{-} (p_n)$ such that $I^{-} (p_n) \subset I^{-} (r_n)$. Given that $\{I^{-} (p_n)\}$ converges to $P$ we have that
$$P = \liminf (I^{-} (p_n)) = \limsup (I^{-} (p_n)).$$
Given $p \in P$, for $n$ large enough, $p_n$ satisfies $p \ll p_n \ll r_n$ and therefore $p \in \liminf (R_n)$ and hence $p \in \limsup (R_n)$. Thus, the structure of the causal boundary implies that $P$ is a maximal IP in $\limsup\{R_n\}$. Therefore, $P \in \hat{L}_{chr}\{R_n\}$ and thus $\{R_n\}$ converges to $P$ in the metric $d_c$ which is a clear contradiction to $R_n \not\in U$. 

Therefore, for the subset $U$ there exists $N \in \mathbb{N}$ such that for $n \geq N$ we have $\hat{I}^{+}(I^{-}(p) \subset U$ as we wanted.
\end{proof}

We now move on into proving that $\hat{M}$ is a non-total imprisoning Lorentzian pre-length space.

\begin{proposition}
Let $M = (a,b) \times_{\alpha} S$ be a globally hyperbolic  warped spacetime with $\int_{c}^{b} \frac{ds}{\sqrt{\alpha (s)}} < \infty$. Then $(\hat{M}, d_c, \hat{\ll},\hat{\leq},\htau)$ is a non-totally imprisoning Lorentzian pre-length space.
\end{proposition}

\begin{proof}
Considering that $\hat{M}$ is an almost strongly causal Lorentzian pre-length space and by Lemma \ref{lemma:lcc}, it is locally causally closed, the proof follows the same argument as the proof of \cite[Thm 3.26 (iii)]{KSlls} using that $\hat{M}$ is locally causally closed and satisfies Proposition \ref{prop:weakstrongcausality} in addition to being $d_c$-compatible.
\end{proof}

\begin{proposition}
Let $M = (a,b) \times_{\alpha} S$ be a globally hyperbolic  GRW spacetime with $\int_{c}^{b} \frac{ds}{\sqrt{\alpha (s)}} < \infty$. Then $(\hat{M}, d_c, \hat{\ll},\hat{\leq},\htau)$ is a globally hyperbolic Lorentzian pre-length space.
\end{proposition}

\begin{proof}
Consider $P,Q \in \hat{M}$ with $P \hat{\leq} Q$, that is, $P \subset Q$. We proceed to show that the causal diamond $\hat{J} (P,Q)$ is compact. If $P = I^{-}(p)$ and $Q = I^{-} (q)$ are both PIPs, there is nothing to prove as $M$ is globally hyperbolic and $I^{-}(p) \hat{\leq} I^{-} (q)$ if and only if $p \leq q$. The case where $P, Q \in \hat{\partial} M$ is not possible since there are no two causally related TIPs.

The only case left is $Q \in \hat{\partial} M$ and $P = I^{-} ((t,r))$ a PIP. Assume that $(t,r) \in Q$ and take a future chain $(t_n, r_{Q})$ generating $Q$, where $r_Q$ is the spatial projection of $Q$ on $S$. Then for some $n$ large enough we have $(t,r) \ll (t_n, r_Q)$, which by definition means
\[ d(s, r_Q) < \int_{t}^{t_n}\frac{dr}{\sqrt{\alpha(r)}} \leq \int_{t}^{b} \frac{dr}{\sqrt{\alpha (r)}} =: L_0. \]
Thus $s \in B_{L_0}^{d} (r_Q) \subset (S,h)$. 

Let $\{R_n \}$ be a sequence of IPs contained in $\hat{J}^{+} (I^{-} ((t,r))) \cap \hat{J}^{-} (Q)$. This sequence cannot contain any terminal set. Therefore $R_{n} = I^{-} ((t_n,r_n))$ and
\[ I^{-} ((t,r)) \subset I^{-} ((t_n,r_n)) \subset Q, \]
for every $n$. This implies that $(t,r) \leq (t_n, r_n)$ and by consequence $d(s,r_n) < L_0$, that is, $r_n \in \overline{B_{L_0}^{d}(s)}$.

By completeness there exists a converging subsequence $r_{n_k} \to w$ with $w \in \overline{B_{L_0}^{d}(r)}$, we omit the subsequence for writing purposes. Then we have two possibilities: 

\begin{enumerate}[label = (\roman*)]
\item $t_n \to t_0 < b$
\item $t_n \nearrow b$
\end{enumerate}

In (i) we have $(t_n , r_n) \to (t_0,w)$. Since $(t,r) \ll (t_n,r_n)$ for every $n$, we have $d (r,w) \leq \int_{t}^{t_0} \frac{ds}{\sqrt{\alpha (s)}}$ which implies $(t,r) \leq (t_0,w)$.

In case (ii), since $I^{-}((t_{n},x_n)) \subset Q:=I^{-}(b,x_Q)$ we have the following integral condition for large $n \in \mathbb{N}$:
\[
	\int_{t_{n}}^{b} \displaystyle{\frac{ds}{\sqrt{\alpha(s)}}} \geq d(x_n,x_Q)
\]
Observe that both the integral condition $\int_{c}^{b} \displaystyle{\frac{ds}{\sqrt{\alpha(s)}}}<\infty$ and the convergence of $\{x_{n}\}$ to $w \in M$ lead to 
\[
0\geq d(\lim_{n\rightarrow} x_{n},x_Q)=d(w,x_{Q}) \geq 0
\]
 Therefore, $w=x_{Q}$ and thus the sequence $\{(t_{n},x_{n})\}$ converges to $(b,x_{Q})$ in the future causal completion and this proves thar $\{I^{-}((t_{n},x_{n}))\}$ converges to $Q \in \hat{J}^{+}(I^{-}(t,r)) \cap Q$. 
\end{proof}

\subsection{Case B. $\int_{c}^{b} \frac{ds}{\sqrt{\alpha (s)}} = \infty$. }

In this section we prove the global hyperbolicity of the future causal completion $\hat{M}$ of a GRW spacetime $M = (a,b) \times_{\alpha} S$ with $(S,h)$ a complete Riemannian manifold and $\int_{c}^{b} \frac{ds}{\sqrt{\alpha(s)}} = \infty$. We further assume that the chronological topology is Hausdorff and therefore coincides with the CLT topology for $\hat{M}$.

We begin with a lemma which describes the behavior of causally related TIPs.

\begin{lemma}\label{lemma:busemanncoinciden}
If $P,Q \in \hat{\partial} M$ with $P \hat{\leq} Q$ then their classes in the Busemann boundary coincide.\footnote{This phenomenon can be interpreted as $\hat{M}$ being a lightlike cone. Note that this is the case for spacetimes when one works with null hypersurfaces: events on the null hypersurface are either on a null generator or spacelike related. Compare also with \cite[Prop. 6.23 (i)]{GromovCauchy}.}
\end{lemma}

\begin{proof}
Let $\alpha,\beta$ be curves in $S$ such that $P = P(b_{\alpha})$ and $Q = P(b_{\beta})$, that is, $P$ and $Q$ are represented by the Busemann functions associated to $\alpha$ and $\beta$, respectively. Assume $[b_{\alpha}] \neq [b_{\beta}]$ and consider
\[
b_{\alpha}(x)-b_{\beta}(x)=k_{2}>0.
\]
Thus, as $\int_{c}^{b} \frac{ds}{\sqrt{\alpha(s)}}=+\infty$, there exists  $t_{0}$ big enough such that
\[
b_{\alpha}(x)> \int_{c}^{t_0} \frac{ds}{\sqrt{\alpha(s)}}>b_{\beta}(x).
\]
Consequently, $(t_0,x) \in P$ and $(t_0,x) \not\in Q$ and the proof is complete.
\end{proof}

\begin{proposition}\label{prop:diamantescompactosinfinito}
The causal diamonds $\hat{J} (P,Q) := \hat{J}^{+} (P) \cap \hat{J}^{-} (Q)$ are compact for every $P,Q \in \hat{M}$.
\end{proposition}

\begin{proof}
If $P$, $Q$ are both proper then $\hat{J} (P,Q) = J^{+} (p) \cap J^{-} (q)$, which is compact due to global hyperbolicity of $M$.

If $P$, $Q$ are both terminal with $\hat{J} (P,Q) \neq \emptyset$ then $P \hat{\leq} Q$, which implies along with Lemma \ref{lemma:busemanncoinciden} that $P$ and $Q$ lie on the same null line on the cone that makes $\hat{\partial} M$. Then $\hat{J} (P,Q)$ is a segment of a null line, which is compact.

Now assume that $P$ is proper and $Q$ is terminal. We know that for any $R \in \hat{J} (P,Q)$ there is a causal curve $\gamma : [0,1] \to \hat{M}$ and $c \in [0,1]$ such that
$$ \gamma (0) = P = I^{-} (p), \quad \gamma (c) = R, \quad \text{and} \quad \gamma (1) = Q.$$
Then $d(P,R) \leq d (P,Q)$ since any causal curve is a distance realizer for $d_c$ (recall Remark \ref{rem:instrinsic}.  Thus for any $R,S \in \hat{J}(P,Q)$ we get by the triangle inequality,
$$d_c (R,S) \leq 2 d_c (P,Q).$$
Thus $\hat{J} (P,Q)$ is bounded. It is also closed since, by Lemma \ref{coro:causally_simple}, both $\hat{J}^+ (P)$ and $\hat{J}^- (Q)$ are closed. 

Given that $\calT _{chr}$, the Gromov and the Busemann completions coincide by \cite[Theorem 5.39]{GromovCauchy} as point sets and topologically. Moreover, the (Cauchy-)Gromov completion satisfies the Heine-Borel property (see \cite[Corollary 4.13]{GromovCauchy}. Thus $\hat{J}(P,Q)$ is compact.
\end{proof}

We proceed to prove almost strong causality for a GRW spacetime with infinite integral and Hausdorff chronological topology.

\begin{proposition}\label{prop:almost_strong_causality_infinito}
Let $M = (a,b) \times_\alpha S$ be a globally hyperbolic GRW spacetime with $\int_{c}^{b} \frac{ds}{\sqrt{\alpha (s)}} = \infty$. Moreover, assume that the chronological topology $\hat{\mathcal{T}}_{chr}$ for $\hat{M}$ is Hausdorff. Then for any $P \in \hat{M}$ and any neighbourhood $U$ of $P$ there exists an open neighbourhood $V \subset U$ of $P$ which is causally convex.
\end{proposition}

\begin{proof}
Similar to Proposition \ref{prop:almost_strong_causality_infinito}, if $P = I^{-} (p)$, global hyperbolicity (therefore strong causality) of $M$ gives us the result. 

Let $P \in \hat{\partial} M$ and $U$ be an open neighbourhood of $P$. If $\{p_n\}$ is a future chain generating $P$ then the sequence $\{I^{-} (p_n)\}$ converges to $P$ with the CLT topology, which coincides with the chronological topology by \cite[Thm 5.3]{clt}. 

Take $Q$ a TIP in $U$ such that $P \subsetneq Q$ (this means $P \hat{<} Q$) and let $q \in Q$ such that $I^{-} (q) \in U, q \notin P$ and $q \notin I^{+} (p_n)$ for $n$ large enough. Note that the sets $B_n = \hat{I}^{+}(I^{-} (p_n)) \cap (\hat{J}^{+} (I^{-}(q)))^{c}$ are causally convex open neighbourhoods of $P$ for every $n$ large enough (recall Example \ref{lemma:j+compcausallyconvex}). We only need to prove that for some large $n$, $B_n \subset U$.

Proceeding by contradiction, we assume that $B_n \nsubset U$ for every $n$. Then for every $n$ there exists $R_n \in B_n$ with $R_n \notin U$. We will prove that $R_n$ converges to some $R \in U$, which is a contradiction to the open character of $U$. Since $R_n \in \hat{I}^{+} (I^{-}(p_n))$ then for every $p \in P$, there exist $p_n$ and $r_n \in R_n$ such that
$$p \ll p_n \ll r_n,$$
which implies that $p \in R_n$ for all $n$ large enough, and $P \subset R_n$ for all $n$ large enough. This means that $P \hat{\leq} R_n$ and, using Lemma \ref{lemma:busemanncoinciden}, $ R_n \hat{\leq} Q$, since otherwise we would have  that $I^{-} (q) \hat{\ll} Q \hat{\leq} R_n$ which contradicts $R_n \in B_n$. Then $R_n \in \hat{J} (P,Q)$ for $n$ large enough and by Proposition \ref{prop:diamantescompactosinfinito} there exists a subsequence of $\{R_n\}$ that converges to $R \in \hat{J} (P,Q) \subset U$. This is a contradiction to the open character of $U$. Thus, $B_n \subset U$ for some large $n$. 
\end{proof}

By combining Propositions \ref{prop:diamantescompactosinfinito} and \ref{prop:almost_strong_causality_infinito}, together with $d_c$-compatibility, give us our main result.

\begin{theorem}
    Let $M = (a,b) \times S$ be a globally hyperbolic GRW spacetime with $\int_{c}^{\infty} \frac{ds}{\sqrt{\alpha (s)}} = \infty$ such that the chronological topology $\hat{\calT}_{chr}$ is Hausdorff. Then $(\hat{M}, d_c, \hat{\ll}, \hat{\leq}, \htau)$ is a globally hyperbolic Lorentzian pre-length space.
\end{theorem}

\section*{Acknowledgments}

L. Ake was supported by Secihti SNII 367994 and CBF-2025-I-2376. D. A. Solis was supported by Secihti SNII 38368 and and CBF-2025-I-2376. S. Burgos acknowledges the support of IMAG-Mar\'{\i}a de Maeztu grant CEX2020-001105-M (funded by MCIN/AEI/10.13039/50110001103) and is also partially supported by the project PID2020-116126GB-I00 (MCIN/AEI/10.13039/501100011033). S. Burgos and D. A. Solis acknowledge the support of BIRS-IMAG to attend the ``Geometry, Analysis, and Physics in Lorentzian Signature" workshop. This work was also partially funded by Austrian Science Fund (FWF) [Grant DOI 10.55776/EFP6].

\bibliographystyle{plain}

\begin{thebibliography}{99}

\bibitem{complls} L. Ake, S. Burgos and D.A. Solis. \emph{Causal completions as Lorentzian pre-length spaces}. Gen. Relat. Gravit., \textbf{54} 108 (2022).

\bibitem{ACS}  L. Ake, A. Cabrera and D.A. Solis. \emph{On the causal hierarchy of Lorentzian length spaces.}  Class. Quantum Grav., \textbf{37}(21):215013 22 (2020).

\bibitem{joniluis} L. Ake and J.L. Flores and J. Herrera. \emph{Causality and c-completion of multiwarped spacetimes}. Class Quantum Grav., \textbf{35}:035014 (2018).

\bibitem{valana} V. Ala\~na and J.L. Flores.\emph{The causal boundary for product spacetimes.} Gen. Relat. Gravit., \textbf{39} 1697--1718 (2007).

\bibitem{solismontes} W. Barrera, L. Montes de Oca, and D. A. Solis. \emph{Comparison theorems for Lorentzian length spaces with lower timelike curvature bounds.} Gen. Relat. Gravit., \textbf{54}(9):107  (2022).

\bibitem{hyper} W. Barrera, L. Montes de Oca, and D. A. Solis. \emph{On the space of compact diamonds of Lorentzian length spaces.} Class. Quantum Grav., \textbf{41}(6):065012 (2024).

\bibitem{beemclt} J.K. Beem. \emph{A metric topology for causally continuous completions.} Gen. Rel. Grav. \textbf{8}(4) 245--257 (1977).

\bibitem{beem} J.K. Beem, P.E. Ehrlich and K.L. Easley. \emph{Global Lorentzian geometry.} Marcel Dekker Inc., New York (1996).


\bibitem{splitting} \textit{T. Beran} et al., {\em The splitting theorem for globally hyperbolic Lorentzian length spaces with non-negative timelike curvature}, Lett. Math. Phys. 113, No. 2, Paper No. 48, 47 p. (2023) 

\bibitem{beran01} T. Beran and C. Sam\"ann. \emph{Hyperbolic angles in Lorentzian length spaces and timelike curvature bounds.} J. London Math. Soc., \textbf{107}(5) 1823--1880  (2023).


\bibitem{timefunctionsLLS} A. Burtscher and L. García-Heveling. {\em Time functions on Lorentzian length spaces}, Ann. Henri Poincaré 26, No. 5, 1533--1572 (2025)



\bibitem{choquet-bruhat_geroch} Y. Choquet-Bruhat and R. Geroch. {\em Global aspects of the Cauchy problem in general relativity}. Commun. Math. Phys. 14, 329–335 (1969)


\bibitem{piotrc0} P.T. Chru\'sciel and J.D.E. Grant, \emph{On Lorentzian causality with continuous metrics}, Class. Quantum Grav., \textbf{29}(14):145001 32 (2012).

\bibitem{bhcausal} I.P. Costa e Silva and J.L. Flores and J. Herrera. \emph{A novel notion of null infinity for c-boundaries and generalized black holes.} J. High Energy Phys.  \textbf{123} (2018).

\bibitem{clt} I.P. Costa e Silva, J.L. Flores and J. Herrera. \emph{Hausdorff closed limits and the c-boundary I: a new topology for the c-completion of spacetimes}. Class. Quantum Grav., \textbf{36} (2019).

\bibitem{ebh} The Event Horizon Telescope Collaboration. \emph{First M87 event horizon telescope results. I. The shadow of the supermassive black hole.} Astrophys. J. Lett., \textbf{875} 1--17 (2019).

\bibitem{flores:revisited} J.L. Flores. \emph{The Causal Boundary of spacetimes revisited}, Comm. Math. Phys., \textbf{276} 611--643 (2007).

\bibitem{floresharris} J.L. Flores, S.G. Harris. \emph{Topology of causal boundary for Standard Static spacetimes}. Class. Quantum Grav., \textbf{24} 2275--2291 (2003).

\bibitem{onthefinal} J.L. Flores, J. Herrera and M. S\'anchez. \emph{On the final definition of the causal boundary and its relation with the conformal boundary}. Adv. Theo. Math. Phys., \textbf{15} 991 -- 1057 (2011).

\bibitem{GromovCauchy} J.L. Flores, J. Herrera and M. Sánchez, \emph{Gromov, Cauchy and causal boundaries for Riemannian, Finslerian and Lorentzian manifolds.}  American Mathematical Society, Providence,  (2013). 


\bibitem{bartnik} J. L. Flores, J. Herrera and D. A. Solis. {\em Low regularity approach to Bartnik's conjecture}, Preprint, arXiv:2412.08967 [math.DG] (2024)



\bibitem{choque-bruhat} Y. Fourès-Bruhat. {\em Théorème d’existence pour certains systèmes d’équations aux dérivées partielles non linéaires}. Acta Math. 88, 141–225 (1952)


\bibitem{gkp} R. Geroch, E. Kronheimer and  R. Penrose. {\emph{Ideal points for space-time}}, {Proc. Roy. Soc. Lond.} A, \textbf{237} 545--567 (1972).

\bibitem{Graf3} M. Graf. \emph{Singularity theorems for $C^1$ Lorentzian metrics}. Commun. Math. Phys., \textbf{378} 1417--1450 (2020).

\bibitem{GKSinext} J.D.E. Grant, M. Kunzinger and C. S\"amann. \emph{Inextendibility of space-times and Lorentzian length spaces}. Ann. Glob. Anal. Geom., \textbf{55} 133--147 (2019).

\bibitem{harrisuniversality} S.G. Harris. \emph{Universality of the future chronological boundary.} J. Math. Phys., \textbf{39} 5427 -- 45 (1998)

\bibitem{harrischrtop} S.G. Harris. \emph{Topology of the future chronological boundary: universality for spacelike boundaries}, Class. Quantum Grav., \textbf{17} 551--603 (2000).

\bibitem{harriscon} S. Harris. \emph{Complete affine connection in the causal boundary: static, spherically symmetric spacetimes}. Gen. Relat. Gravit. \textbf{49} (2017).

\bibitem{HE} S.W. Hawking and G.F.R. Ellis. \emph{The Large Scale Structure of Space-Time}. Cambridge University Press, (1973)

\bibitem{hendike} J. Hedicke. \emph{Lorentzian distance functions in contact geometry.} J. Topol. Anal.,  1--21 10.1142/S179352532250008X (2022)

\bibitem{kpcs} E. Kronheimer and R. Penrose. \emph{On the structure of causal spaces.} Math. Proc. Cambridge Phil. Soc. \textbf{63}(2) 481--502 (1967) 

\bibitem{Kun3} M. Kunzinger, A. Ohanyan, B. Schinnerl, and R. Steinbauer. \emph{The Hawking-Penrose singularity theorem for $C^1$  Lorentzian metrics.} Comm. Math. Phys., \textbf{391} (3)  1143--1179 (2022).

\bibitem{KSlls} M. Kunzinger and C. S\"amann, \emph{Lorentzian length spaces}. Ann. Glob. Anal. Geom., \textbf{54} (3)  399--447 (2008).

\bibitem{ligo} LIGO Scientific Collaboration Virgo Collaboration. \emph{Observation of gravitational waves from a binary black hole merger.} Phys. Rev. Lett., \textbf{116} 061102  (2016).

\bibitem{ling01} E. Ling. \emph{Aspects of $C^{0}$ causal theory.} Gen. Relat. Gravit., \textbf{52} 57 (2020).

\bibitem{Mcann}
 R. McCann and C. S\"amann. \emph{A Lorentzian analog for Hausdorff dimension and measure}. Pure  App. Anal. \textbf{4}(2) 367–-400 (2022). 


\bibitem{mingucone} E. Minguzzi. \emph{Causality theory for closed cone structures with applications.} Rev. Math. Phys., \textbf{31}(5):1930001, 139 (2019).



\bibitem{lorentzmetric} E. Minguzzi and S. Suhr, \emph{Lorentzian metric spaces and their Gromov-Hausdorff convergence}. Lett. Math. Phys., \textbf{114}  73 (2024). 


\bibitem{mondino-samann} A. Mondino and C. S\"amann, {\em Lorentzian Gromov-Hausdorff convergence and pre-compactness}, Preprint, arXiv:2504.10380 [math.DG] (2025)


\bibitem{Olaf} O. M\"uller. \emph{Topologies on the future causal completion}. arXiv:1909.03797 (2022).

\bibitem{Olaf2} O. M\"uller. \emph{Functors in Lorentzian geometry: three variations on a theme} Gen. Relat. Grav., \textbf{55}(2) 39 (2023).


\bibitem{zero} R. Penrose. \emph{Zero rest-mass fields including gravitation: asymptotic behaviour}. {Proc. Roy. Soc. Lon.}, \textbf{284}(1397)  159--203, (1965).

\bibitem{penrose_sing} R. Penrose. {\em Gravitational collapse and space-time singularities}. Phys. Rev. Lett. {\bf 14}, 57–59 (1965)

\bibitem{surya} S. Surya. \emph{The causal set approach to quantum gravity.}  Living Rev. Relativ., \textbf{22} 5 (2019).

\end{thebibliography}

\bigskip

\noindent\textsc{Luis Ake Hau.}\\
Instituto Tecnol\'ogico Superior de Valladolid.\\
97780 Valladolid, M\'exico.\\
luis.ah@valladolid.tecnm.mx

\medskip

\noindent\textsc{Saul Burgos.}\\
Departamento de Geometr\'ia y Topolog\'ia, Universidad de Granada.\\
Instituto de Matemáticas IMAG
Granada, Spain.\\
sburgos@ugr.es

\medskip

\noindent\textsc{Didier A. Solis.}\\
Facultad de Matem\'aticas, Universidad Aut\'onoma de Yucat\'an.\\
97203, M\'erida, M\'exico.\\
didier.solis@correo.uady.mx

\end{document}